\documentclass[aps,prd,preprint,a4paper,superscriptaddress,nofootinbib]{revtex4-1}

\usepackage{amsmath,mathrsfs,amssymb,amsthm}
\usepackage{braket}
\usepackage{fancyhdr}
\usepackage[colorlinks=true]{hyperref}

\usepackage{mathtools}
\usepackage{here}
\usepackage{color}
\usepackage{graphicx}
\usepackage{enumerate}
\usepackage[normalem]{ulem}
\usepackage{tensor}
\newtheorem{proposition}{Proposition}

\newcommand{\nn}{\nonumber}
\newcommand{\mo}[2]{{\mathcal O}(#1^{#2})}
\newcommand{\dd}{\mathrm{d}}

\begin{document}
\title{A Lie algebra based approach to asymptotic symmetries in general relativity}
\author{Takeshi Tomitsuka}
\affiliation{Graduate School of Science, Tohoku University, Sendai, 980-8578, Japan}
\author{Koji Yamaguchi}\thanks{Current Affiliation : Graduate School of Informatics and Engineering, The University of Electro-Communications, 1-5-1 Chofugaoka, Chofu,
Tokyo 182-8585, Japan}
\affiliation{Graduate School of Science, Tohoku University, Sendai, 980-8578, Japan}
\author{Masahiro Hotta}
\affiliation{Graduate School of Science, Tohoku University, Sendai, 980-8578, Japan}
\begin{abstract}
Asymptotic symmetries of black hole spacetimes have received much attention as a possible origin of the Bekenstein-Hawking entropy in black hole thermodynamics.
    In general, it takes hard efforts to find appropriate asymptotic conditions on a metric and a Lie algebra generating the transformation of symmetries with which the corresponding charges are integrable.
    We here propose an alternative approach to construct building blocks of asymptotic symmetries of a given spacetime metric.
    Our algorithmic approach may make it easier to explore asymptotic symmetries in any spacetime than in conventional approaches.
    As an explicit application, we analyze the asymptotic symmetries on Rindler horizon.
    We find a new class of symmetries related with dilatation transformations in time and in the direction perpendicular to the horizon, which we term superdilatations.
\end{abstract}
\maketitle
\section{Introduction}
Since Hawking radiation can be emitted out of a black hole \cite{hawking_1974_black_hole_explosions}, it is
widely believed that a black hole carries the Bekenstein-Hawking (BH) entropy
$A/(4G),$ where $A$ is the area of the horizon and $G$ is the gravitational constant.
The origin of the BH entropy has been explored for a long time from various points of view.
In field theory, the BH entropy is suggested to be derived from quantum entanglement \cite{Bombelli_1986_entropy_for_BH,Srednicki_1993_entropy_area}.
It is also pointed out that entanglement may be the origin of the BH entropy in quantum gravity
\cite{Susskind_1994,Fiola_1994,Emparan_2006,Azeyanagi_2008}.
Besides, in string theory, special D-branes correspond to extremal black holes in the classical regime. The value of logarithm of the number of BPS states of the branes approaches the value of its corresponding BH entropy \cite{Strominger_1996}.

Recently, soft hair at the horizon \cite{Hotta_2001,Hotta_2002,Hawking_2016} has attracted much
attention as a possible origin of the BH entropy \cite{Afshar_2016,Mirbabayi_2016,Hotta_2016,Mao_2017,Ammon_2017,Bousso_2017,Hotta_2018,Chu_2018,Haco_2018,Raposo_2019,Grumiller_2020,Averin_2020}. In a near-horizon region of a black hole, asymptotic symmetries
emerge and generate microstates which contribute to the BH entropy in the
standard way of statistical mechanics. In 2001, supertranslation and
superrotation with non-vanishing charges were discovered as horizon asymptotic
symmetries of a Schwarzschild black hole in $(1+3)$-dimensional general relativity
\cite{Hotta_2001,Hotta_2002}. Supertranslation is time translation depending on the
position at the horizon, while superrotation is a 2-dimensional general coordinate
transformation on the horizon.
In 2016, Hawking, Perry and Strominger rediscovered the symmetries and named the micro states generated by the
transformations as soft hair \cite{Hawking_2016}.
Their work has stimulated interest in the quest for other symmetries at the horizon \cite{Mao_2017,Grumiller_2020}.

Exploration of asymptotic symmetries near a boundary like a horizon is accompanied by hard effort.
In the first stage, we fix an asymptotic condition of metric components near the boundary.
In the second stage, we solve asymptotic Killing equations for the metric components so that the asymptotic behavior of the metric is preserved under diffeomorphisms generated by vector fields.
In the third stage, we check whether the charges associated with the diffeomorphisms are integrable.
If the charges are not integrable, it is required to repeat the above three stages until an appropriate asymptotic condition is found.
In the fourth stage, if the charges satisfy the integrability condition, we should finally check whether the charges take various values for solutions of the Einstein equations.
At this stage, it often happens that all the charges vanish, implying that all the diffeomorphisms we have selected may be gauge freedoms.
In this case, to find non-vanishing charges, we restart from the first stage.
Although there are several ways to construct a charge in general relativity, such as the Regge-Teitelboim method \cite{REGGE_1974} and the covariant phase method \cite{Lee_Wald_1990,Wald_1993,Iyer_Wald_1994,Iyer_Wald_1995,Wald_Zoupas_2000} developed by Iyer, Lee, Wald and Zoupas, all of them require such efforts in trials and errors. See also Ref.~\cite{Kijowski} for early study related to this method. 

In this paper, we take a shorter route to find non-trivial asymptotic properties and propose an approach without imposing asymptotic behaviors of metrics by hand.
For a given background metric $\bar{g}_{\mu\nu}$ of interest, we consider a set of metrics which are diffeomorphic to it so that purely gravitational properties of asymptotic symmetries can be analyzed.
In this case, some of these symmetries cannot be gauged away.
For example, a diffeomorphism associated with a Lorentz boost is not a gauge freedom since it changes energy and momentum of a black hole.
As a guiding principle to find a non-trivial diffeomorphism, at the first stage, we adopt a condition under which the charges take non-vanishing values for some metrics generated by the diffeomorphism.
Analyzing this condition at the background metric, we can find the candidates for vector fields generating a non-trivial diffeomorphism.
A key advantage of our protocol is the fact that the diffeomorphism generated by these vector fields cannot be gauged away by construction.
As a consequence, it helps us to find a minimal non-trivial diffeomorphism as a building block of asymptotic symmetries.
After identifying the minimal Lie algebra $\mathcal{A}$ spanned by the vector fields and their commutators, we first check the integrability condition at the background metric $\bar{g}$.
Our approach to find a non-trivial diffeomorphism satisfying the integrability condition at the background metric may reduce the difficulties in trials and errors in the conventional approach.
Finally, we check the integrability condition for a set of metrics connected to the background metric by diffeomorphisms generated by the Lie algebra $\mathcal{A}$.
If the integrability condition is satisfied, the charges can be calculated as an integral along a path from the background metric to other metrics.

To demonstrate our approach, we investigate asymptotic symmetries on the Rindler horizon in $(1+3)$-dimensional Rindler spacetime.
We derive a general condition for vector fields that generate diffeomorphisms and along which the variations of the corresponding charges do not vanish at the background metric.
We show that supertranslations and superrotations on the Rindler horizon generate diffeomorphisms which cannot be gauged away, confirming the result in prior research \cite{Hotta_2016}.
Furthermore, we find a new class of non-trivial diffeomorphisms, which we term superdilatation.
This superdilatation includes two classes of diffeomorphisms.
One of them is an extension of dilatation in the direction perpendicular to the horizon.
The other is an extension of dilatation in the time direction.
We explicitly calculate the expression of charges for an example of the superdilatation algebra.

This paper is organized as follows:
In Sec.~\ref{sec:conventional}, we briefly review a conventional approach requiring much effort in trial and error.
In Sec.~\ref{sec:Wald_method}, we briefly review the covariant phase space method, which is adopted in this paper.
In Sec.~\ref{sec:setup}, we explain our approach to construct a building block of asymptotic symmetries.
In Sec.~\ref{sec:asymptotic_sym_in_Rindler}, we find a new symmetry on the Rindler horizon called superdilatation by using our approach.
In Sec.~\ref{sec:summary}, we present the summary of this paper.
In this paper, we set the speed of light to unity: $c=1$.

\section{A conventional approach dependent on luck}
\label{sec:conventional}
A standard approach to explore the asymptotic symmetries requires setting the asymptotic form of metrics near the boundary. The success of exploration severely depends on this metric setting. If an inappropriate metric is chosen, then we completely fail to find the non-trivial symmetries. If we have a deep insight to fix the metric, the non-trivial symmetries appear in the theory. In order to explain this situation, we first make a brief review of the conventional approach with the canonical method \cite{REGGE_1974}.

For example, the authors in Ref.~\cite{brown1986} analyzed asymptotic symmetries in $(1+2)$-dimensional asymptotic anti-de Sitter (AdS) spacetime.
The background metric $\bar{g}_{\mu\nu}$ is given by
\begin{align}
    \left(
    \begin{array}{ccc}
        \bar{g}_{tt}     & \bar{g}_{tr}     & \bar{g}_{t\phi}    \\
        \bar{g}_{rt}     & \bar{g}_{rr}     & \bar{g}_{r\phi}    \\
        \bar{g}_{\phi t} & \bar{g}_{\phi r} & \bar{g}_{\phi\phi}
    \end{array}
    \right) & = \left(
    \begin{array}{ccc}
        -\left(\frac{r^{2}}{l^{2}} + 1 \right) & 0                                          & 0     \\
        0                                      & \left(\frac{r^{2}}{l^{2}} + 1 \right)^{-1} & 0     \\
        0                                      & 0                                          & r^{2}
    \end{array}
    \right),
\end{align}
where $l = (-1/\Lambda)^{1/2}$.
It describes the exact AdS metric which is a solution of the Einstein equations with negative cosmological constant $\Lambda$.
The exact AdS spacetime has six Killing vectors, thus the goal of exploration of the asymptotic symmetries is to get at least six asymptotic Killing vectors. The AdS boundary is located at $r = \infty$.
Near the AdS boundary, we set the asymptotic form of the metric as
\begin{align}
    g_{\mu\nu} & = \bar{g}_{\mu\nu} + \delta g_{\mu\nu}    \label{set_of_metric_with_asy_behav_1}.
\end{align}
Let us consider two forms of the metric. One of them is the following ansatz:
\begin{align}
    \left(\delta g_{\mu\nu}\right) & = \left(
    \begin{array}{ccc}
        0                                   & 0 & A\left(\frac{r^{2}}{l^{2}}+1\right)     \\
        0                                   & 0 & 0                                       \\
        A\left(\frac{r^{2}}{l^{2}}+1\right) & 0 & A^{2}\left(\frac{r^{2}}{l^{2}}+1\right)
    \end{array}\label{eq_asy_1}
    \right),\quad  (|A| <|l|).
\end{align}
It can be shown that the vector field preserving the above metric is given by a linear combination of $\partial_{t}$ and $\partial_{\phi}$, which is denoted by $\xi$. Thus, in this case, we have only two asymptotic Killing vectors. 
The variation of associated charges $J[\xi]$ is
\begin{align}
    \delta J[\xi] = 4\pi\xi^{\phi}\delta A.
\end{align}
The charges are integrable and calculated as
\begin{align}
    J[\partial_{t}]    & =0,       \\
    J[\partial_{\phi}] & = 4\pi A,
\end{align}
where the integral constants are chosen such that $J[\xi]=0$ at the AdS spacetime.
In order to get more than one non-vanishing charge, we should replace Eq.~\eqref{eq_asy_1} with another form.
A successful one is the following:
\begin{align}
    \left(\delta g_{\mu\nu}\right) & = \left(
    \begin{array}{ccc}
            {\mathcal O}(1)      & {\mathcal O}(r^{-3}) & {\mathcal O}(1)      \\
            {\mathcal O}(r^{-3}) & {\mathcal O}(r^{-4}) & {\mathcal O}(r^{-3}) \\
            {\mathcal O}(1)      & {\mathcal O}(r^{-3}) & {\mathcal O}(1)
        \end{array}
    \right).
    \label{eq_asy_2}
\end{align}
The solution of the asymptotic Killing equation is given by
\begin{align}
    \xi =\left(
    \begin{array}{c}
        \xi^{t} \\
        \xi^{r} \\
        \xi^{\phi}
    \end{array}
    \right)=\left(
    \begin{array}{c}
        lT(t,\phi)+\frac{l^{3}}{r^{2}}\overline{T}(t,\phi) + {\mathcal O}(r^{-4}) \\
        rR(t,\phi) + {\mathcal O}(r^{-1})                                         \\
        \Phi(t,\phi) + \frac{l^{2}}{r^{2}}\overline{\Phi}(t,\phi) + {\mathcal O}(r^{-4})
    \end{array}
    \right),
\end{align}
where the functions $T(t,\phi),\overline{T}(t,\phi),R(t,\phi),\Phi(t,\phi)$ and $\overline{\Phi}(t,\phi)$ satisfy
\begin{align}
    l\partial_{t}T(t,\phi) & = \partial_{\phi}\Phi(t,\phi) = -R(t,\phi),\ \partial_{\phi}T(t,\phi) = l\partial_{t}\Phi(t,\phi),   \\
    \overline{T}(t,\phi)   & = -\frac{l}{2}\partial_{t}R(t,\phi),\ \overline{\Phi}(t,\phi) = \frac{1}{2}\partial_{\phi}R(t,\phi).
\end{align}
It is shown that the charges associated with the vector fields are integrable. Surprisingly, the algebra of the charges is a direct sum of two Virasoro algebras which are infinite dimensional Lie algebras in contrast to the first case.
Unfortunately, however, there is no systematic way to find such a successful asymptotic form in Eq.~\eqref{eq_asy_2}.

The conventional approach is shown schematically in FIG.~\ref{flowchart_1}.
In the first step, we determine an asymptotic form of the metric near the boundary.
In the second step, we solve asymptotic Killing equations for the metric components so that the asymptotic form of the metric is preserved under diffeomorphisms generated by vector fields.
In the third step, we check whether the charges associated with the diffeomorphisms are integrable.
If the charges are not integrable, we have to repeat the above three steps until we successfully find an appropriate asymptotic condition.
In the fourth step, if the charges are integrable, we check whether they take various values for solutions of the Einstein equations.
If they do, we obtain non-trivial charges. However, if not, we have to restart from the first step since all the diffeomorphisms generated by the vector fields we have selected are gauge freedom.
Such a failure often happens in the conventional approach.
As we have seen, we have to determine the asymptotic form of metric by trials and errors. It usually takes much efforts and might turn out not to serve the purpose in the end.

So far, we gave a review of conventional approach with the canonical method.
The same approach has been taken in studies of asymptotic symmetries using the covariant phase space method.
In other words, the flow chart in Fig.~\ref{flowchart_1} is often adopted in the covariant phase space method, e.g., in Refs.~\cite{Hollands_2005,Ishibashi_2005}.
In the next section, we introduce the covariant phase space method which we adopt in this paper. In Sec.~\ref{sec:setup}, we explain our approach which may reduce the above efforts.
\newpage
\begin{figure}[H]
    \includegraphics[width=15cm,keepaspectratio]{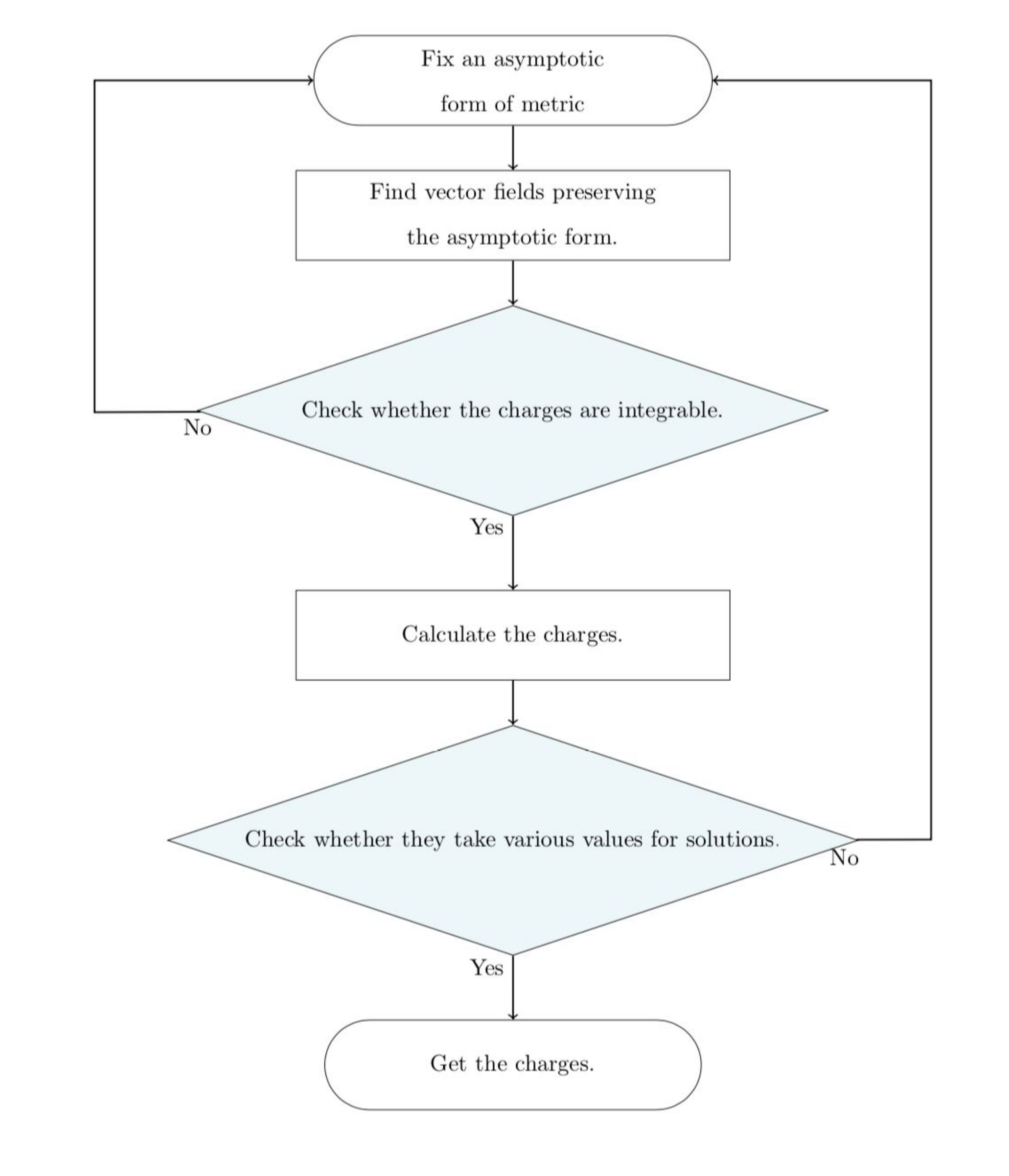}
    \caption{A flow chart of the conventional approach.}
    \label{flowchart_1}
\end{figure}
\newpage
\section{A brief review on a covariant phase space method}\label{sec:Wald_method}
In Sec.~\ref{sec:conventional}, we briefly reviewed the canonical method to explore the asymptotic symmetries.
In this paper, we will use the covariant phase space method \cite{Lee_Wald_1990,Wald_1993,Iyer_Wald_1994,Iyer_Wald_1995,Wald_Zoupas_2000}. An advantage of the method is covariant calculation independent of local coordinates without using the Arnowitt-Deser-Misner (ADM) decomposition \cite{ADM_1959}. Here let us briefly review the covariant phase space method to calculate the charge corresponding to a diffeomorphism. Although the covariant phase space method can be applied to all diffeomorphism invariant theories, we focus on the Einstein gravity.

Consider the Einstein-Hilbert action
\begin{align}
    S & = \int_{\mathcal M}\dd^{d}x {\mathcal L}_{EH},
\end{align}
where the Lagrangian density is given by ${\mathcal L}_{EH} \coloneqq \frac{1}{16\pi G}\sqrt{-g}R$,  $\int_\mathcal{M}d^dx$ denotes the integral over a $d$-dimensional spacetime $\mathcal{M}$, $g$ and $R$ are the determinant of the metric $g_{\mu\nu}$ and the Ricci scalar, respectively.
The variation of ${\mathcal L}_{EH}$ is given by
\begin{align}
    \delta {\mathcal L}_{EH} & = -\frac{\sqrt{-g}}{16\pi G}G^{\mu\nu}\delta g_{\mu\nu} + \partial_{\mu}\Theta^{\mu}(g,\delta g),
\end{align}
where $G_{\mu\nu}$ is the Einstein tensor and $\Theta$ is the pre-symplectic potential defined by
\begin{align}
    \Theta^{\mu}(g, \delta g) & = \frac{\sqrt{-g}}{16\pi G}\left(g^{\mu\alpha}\nabla^{\beta}\delta g_{\alpha\beta} - g^{\alpha\beta}\nabla^{\mu}\delta g_{\alpha\beta}\right).
    \label{eq:presymplectic_potential}
\end{align}
In the following, for notational symplicity, the metric $g_{\mu\nu}$ is abbreviated as $g$ in the arguments of functions.

The Einstein-Hilbert action is invariant under the Lie derivative along an arbitrary vector field $\xi$ up to a total derivative term. Therefore, for an infinitesimal transformation of the metric $\delta_\xi g_{\mu\nu}=\pounds_\xi g_{\mu\nu}$ where $\pounds_{\xi}$ represents the Lie derivative with respect to $\xi$, the corresponding Noether current is given by
\begin{align}
    J^{\mu}[\xi] :=\Theta^{\mu}(g, \pounds_{\xi}g) -\xi^{\mu}{\mathcal L}_{EH},
\end{align}
which satisfies
\begin{align}
    \partial_\mu J^{\mu}[\xi]=\frac{\sqrt{-g}}{16\pi G} G^{\mu\nu}\pounds_{\xi}g_{\mu\nu}.
\end{align}
For a solution $g_{\mu\nu}$ of the Einstein equations, the current is conserved:
\begin{align}
    \partial_\mu J^\mu[\xi]\approx 0,
\end{align}
where $\approx$ means that the equality holds for any solution of the equation of motion, i.e., the Einstein equations.
By using the Poincar\'e lemma, there exists a 2-form $Q^{\mu\nu}[\xi]$ of the spacetime  satisfying
\begin{align}
    \label{Noether}
    J^\mu[\xi]\approx \partial_{\nu}Q^{\mu\nu}[\xi].
\end{align}
More generally, as shown in the Appndix of \cite{Iyer_Wald_1995}, we have
\begin{align}
J^{\mu}[\xi] = \partial_{\nu}Q^{\mu\nu}[\xi] + \mathcal{C}\indices{^{\mu}_{\nu}}\xi^{\nu},
\end{align}
where $\mathcal{C}\indices{^{\mu}_{\nu}}$ is a constraint satisfying $\mathcal{C}\indices{^{\mu}_{\nu}} \approx 0$.  
In the case of Einstein gravity, the 2-form is given by
\begin{align}
    Q^{\mu\nu}[\xi] = -\frac{\sqrt{-g}}{8\pi G}\nabla^{[\mu}\xi^{\nu]}
    \label{eq:Komar},
\end{align}
while $C\indices{^\mu_\nu}$ is given by
\begin{align}
	\mathcal{C}\indices{^{\mu}_{\nu}} = \frac{\sqrt{-g}}{8\pi G}G\indices{^{\mu}_{\nu}},
\end{align}
where the bracket $[\ \ ]$ for indices is an anti-symmetric symbol defined as
\begin{align}
	A_{[\mu_{1} \cdots \mu_{d}]} \coloneqq \frac{1}{d!}\sum_{\sigma \in S_{d}}(-1)^{\sigma}A_{\mu_{\sigma(1)}\cdots\mu_{\sigma(d)}},
\end{align}
where $S_{d}$ is a permutation group.
The corresponding Noether charge of $\xi$ is given by
\begin{align}
    Q[\xi] & \coloneqq \int_{\Sigma}(\dd^{d-1}x)_{\mu}J^{\mu}[\xi] \nonumber                        \\
             & \approx \int_{\Sigma}(\dd^{d-1}x)_{\mu}\partial_{\nu}Q^{\mu\nu}[\xi] \nonumber \\
             & =\oint_{\partial\Sigma}(\dd^{d-2}x)_{\mu\nu}Q^{\mu\nu}[\xi],
    \label{Q_xi}
\end{align}
where $\Sigma$ is a $(d-1)$-dimensional submanifold embedded in ${\mathcal M}$, $\partial\Sigma$ is the boundary of $\Sigma$ and the integral measure is defined as
\begin{align}
    (\dd^{d-p}x)_{\mu_{1}\dots\mu_{p}} \coloneqq \frac{\epsilon_{\mu_{1}\dots\mu_{p}\mu_{p+1}\dots\mu_{d}}}{d!(d-p)!}\dd x^{\mu_{p+1}}\wedge\dots\wedge\dd x^{\mu_{d}}.
    \label{int_measure}
\end{align}
In Eq.~\eqref{int_measure}, $\epsilon_{{\mu_{1}}\cdots\mu_{d}}$ is the $d$-dimensional Levi-Civita symbol defined as
\begin{align}
\epsilon_{\mu_{1}\cdots\mu_{d}} &= \epsilon_{[\mu_{1}\cdots\mu_{d}]} \\
\epsilon_{1\cdots d} &=1.
\end{align}
In the third line in Eq.~\eqref{Q_xi}, we have used Stokes' theorem.

Let $\delta_1 g$ and $\delta_2 g$ be arbitrary linearized perturbations of metric $g$ in question. Let $\delta_i f(g)$ denote the variation of a function $f(g)$ with respect to each perturbation $\delta_i g$. 
With these notations, the pre-symplectic current is defined by 
\begin{align}
    \omega^{\mu}(g, \delta_{1}g, \delta_{2}g) := \delta_{1}\Theta^{\mu}(g, \delta_{2}g) - \delta_{2}\Theta^{\mu}(g,\delta_{1}g).
\end{align}
We further define the pre-symplectic form $\Omega(g, 
\delta_{1}g, \delta_{2}g)$ as
\begin{align}
    \Omega(g, \delta_{1}g, \delta_{2}g) := \int_{\Sigma}(\dd^{d-1}x)_{\mu}\omega^{\mu}(g,\delta_{1}g, \delta_{2}g),
\end{align}
which is a 2-form on the field configuration space. 

Let $H[\xi]$ denote the charge which generates an infinitesimal transformation along a vector field $\xi$. The variation of the charge with respect to an arbitrary perturbation $\delta g$ is given by \cite{Lee_Wald_1990,Wald_1993,Iyer_Wald_1994,Iyer_Wald_1995,Wald_Zoupas_2000}
\begin{align}
    \delta H[\xi] & = \Omega(g, \delta g, \pounds_{\xi}g) =\int_{\Sigma}(\dd^{d-1}x)_{\mu}\omega^{\mu}(g,\delta g, \pounds_{\xi}g).
\end{align}
The variation of the Noether current can be recast into
\begin{align}
    \delta J^{\mu}[\xi] \approx \omega^{\mu}(g, \delta g, \pounds_{\xi}g) -\partial_{\nu}[2\xi^{[\mu}\Theta^{\nu]}(g,\delta g)]\label{eq_omega},
\end{align}
where $g_{\mu\nu}$ is assumed to be the solution of the Einstein equations, while $\delta g_{\mu\nu}$ does not necessarily satisfy the linearized Einstein equations. Equation~\eqref{eq_omega} can be rewritten as
\begin{align}
    \omega^\mu(g,\delta g, \pounds_\xi g) \approx \delta\mathcal{C}\indices{^{\mu}_{\nu}}\xi^{\nu}+\partial_\nu S^{\mu\nu}\left(g,\delta g,\pounds_\xi g\right),
\end{align}
where we have defined
\begin{align}
    S^{\mu\nu}\left(g,\delta g,\pounds_\xi g\right) & \coloneqq \delta Q^{\mu\nu}[\xi]+2\xi^{[\mu}\Theta^{\nu]}(g,\delta g)\nonumber \\
                                                    & = \frac{\sqrt{-g}}{8\pi G}\left(
    -\frac{1}{2}\delta g^{\alpha}_{\ \alpha}\nabla^{[\mu}\xi^{\nu]} + \delta g^{\alpha[\mu}\nabla_{\alpha}\xi^{\nu]} - \nabla^{[\mu}\delta g^{\nu]\alpha}\xi_{\alpha} + \xi^{[\mu}\nabla_{\alpha}\delta g^{\nu]\alpha} - \xi^{[\mu}\nabla^{\nu]}\delta g^{\alpha}_{\ \alpha}\right).
    \label{eq_definition_S}
\end{align}
Thus, if  $H[\xi]$ exists, it satisfies
\begin{align}
    \delta H[\xi] \approx \int_{\Sigma}(\dd^{d-1}x)_{\mu}\delta \mathcal{C}\indices{^{\mu}_{\nu}}\xi^{\nu} + \oint_{\partial\Sigma}(\dd^{d-2}x)_{\mu\nu}S^{\mu\nu}(g,\delta g, \pounds_{\xi}g).
    \label{eq:H_for_on_shell}
\end{align}
When $\delta g_{\mu\nu}$ is a solution of the linearized Einstein equation, $\delta \mathcal{C}\indices{^\mu_\nu}=0$. Therefore, we get
\begin{align}
    \delta H[\xi] \approx  \oint_{\partial\Sigma}(\dd^{d-2}x)_{\mu\nu}S^{\mu\nu}(g,\delta g, \pounds_{\xi}g).
\end{align}
Since $H[\xi]$ does not always exist, we need to impose an additional condition for $g_{\mu\nu}$ and $\xi$ to obtain the charge $H[\xi]$, which is referred to as the integrability condition.
In Ref.~\cite{Wald_Zoupas_2000}, the integrability condition is introduced. 

As a simplest example, let us first consider whether the charges are integrable for a set of solutions of the Einstein equation $g_{\mu\nu}(\lambda_1,\lambda_2)$, which is smoothly parameterized by two real parameters $\lambda_1$ and $\lambda_2$. 
A linearized perturbation $\delta_i g_{\mu\nu}(\lambda_1,\lambda_2) $ is defined by $\delta_i g_{\mu\nu}(\lambda_1,\lambda_2)\coloneqq \frac{\partial }{\partial \lambda_i}g_{\mu\nu}(\lambda_1,\lambda_2)$. 
If $H[\xi]$ exists, due to the equality of mixed partial derivatives, we have
\begin{align}
    0&=\left(\frac{\partial}{\partial \lambda_{1}}\frac{\partial}{\partial \lambda_{2}}-\frac{\partial}{\partial \lambda_{2}}\frac{\partial}{\partial \lambda_{1}}\right)H[\xi]\bigg|_{g(\lambda_1,\lambda_2)}\nonumber\\
    &=(\delta_{1}\delta_{2}-\delta_{2}\delta_{1})H[\xi]|_{g(\lambda_1,\lambda_2)}.
\end{align}
As long as there is no topological obstruction, this is a necessary and sufficient condition for the charge $H[\xi]$ to exist.

Similarly, for a general set of solutions $g_{\mu\nu}$ of the Einstein equations, for $H[\xi]$ to exist, it must holds
\begin{align}
    \label{condition1}
    0 & = (\delta_{1}\delta_{2}-\delta_{2}\delta_{1})H[\xi]\nonumber                                                                                                           \\
      & = -\int_{\partial \Sigma}(\dd^{d-2}x)_{\mu\nu}\left(\xi^{[\mu}\delta_{1}\Theta^{\nu]}(g,\delta_{2}g)-\xi^{[\mu}\delta_{2}\Theta^{\nu]}(g,\delta_{1}g)\right) \nonumber \\
      & =-\int_{\partial\Sigma}(\dd^{d-2}x)_{\mu\nu}\xi^{[\mu}\omega^{\nu]}(g,\delta_{1}g,\delta_{2}g)\nonumber                                                                \\
      & \approx-\int_{\partial\Sigma}(\dd^{d-2}x)_{\mu\nu}\xi^{[\mu}\partial_\alpha S^{\nu]\alpha}(g,\delta_{1}g,\delta_{2}g)
\end{align}
for arbitrary linearized perturbations $\delta_1 g$ and $\delta_{2}g$ of the metric in question, 
where we have used Eq.~\eqref{eq:H_for_on_shell}.
This is a necessary condition for $H[\xi]$ to exist.
It is also a sufficient condition if the space of $g_{\mu\nu}$ does not have any topological obstruction \cite{Wald_Zoupas_2000}.
Shifting the charge by a constant, it is always possible to make the charges vanish at a reference metric $g^{(0)}_{\mu\nu}$.
By using a smooth one-parameter set of solutions $g_{\mu\nu}(\lambda)$ such that $g_{\mu\nu}(0)=g^{(0)}_{\mu\nu}$ and $g_{\mu\nu}(1)=g_{\mu\nu}$, the charge is given by
\begin{align}
    H[\xi] = \int_{0}^{1}\dd\lambda\int_{\partial \Sigma}(\dd^{d-2}x)_{\mu\nu}(\partial_{\lambda} Q^{\mu\nu}[\xi](g,\partial_{\lambda}g) + 2\xi^{[\mu}\Theta^{\nu]}(g,\partial_{\lambda} g)).\label{eq_charge_int_along_path}
\end{align}
Note that the charge defined in Eq.~\eqref{eq_charge_int_along_path} is independent of the choice of the path $g_{\mu\nu}(\lambda)$ as long as Eq.~\eqref{condition1} is satisfied. 

In the later sections, we adopt this method.

\section{Our approach}
\label{sec:setup}
In this section, we explain our approach to explore the asymptotic symmetries.
A guiding principle is proposed to determine $\delta g_{\mu\nu}$. The choice of $\delta g_{\mu\nu}$ ensures us to obtain the non-trivial charges of the asymptotic symmetries as long as the integrability of the charges is satisfied. As a consequence, we can get the diffeomorphisms which cannot be gauged away.

We consider a Lie algebra $\mathcal{A}$ of vector fields, and a set of metrics connected to the fixed background metric $\bar{g}_{\mu\nu}$, which is a solution of the Einstein equations, by all the diffeomorphism generated by $\mathcal{A}$.
For an arbitrary variation $\delta$ and an arbitrary element $g_{\mu\nu}$ of this set, there exists a vector field $\chi$ in the algebra such that
\begin{align}
    \label{set2}
    \delta g_{\mu\nu} = \pounds_{\chi} g_{\mu\nu}.
\end{align}
With this set of metrics, we can analyze the properties of the background metric $\bar{g}_{\mu\nu}$ since all the metrics are diffeomorphic to it.
It should be noted that as opposed to the conventional approach, we do not need to check whether the variation of the metric satisfies the linearized Einstein equations since the Einstein equations are invariant under diffeomorphisms.
Hereafter, $g_{\mu\nu}$ denotes a metric connected to $\bar{g}_{\mu\nu}$ via a diffeomorphism generated by the Lie algebra $\mathcal{A}$.
A schematic picture of the set of metric is shown in FIG.~\ref{fig:configuration_space}.
\begin{figure}[htbp]
    \centering
    \includegraphics[width=7.5cm]{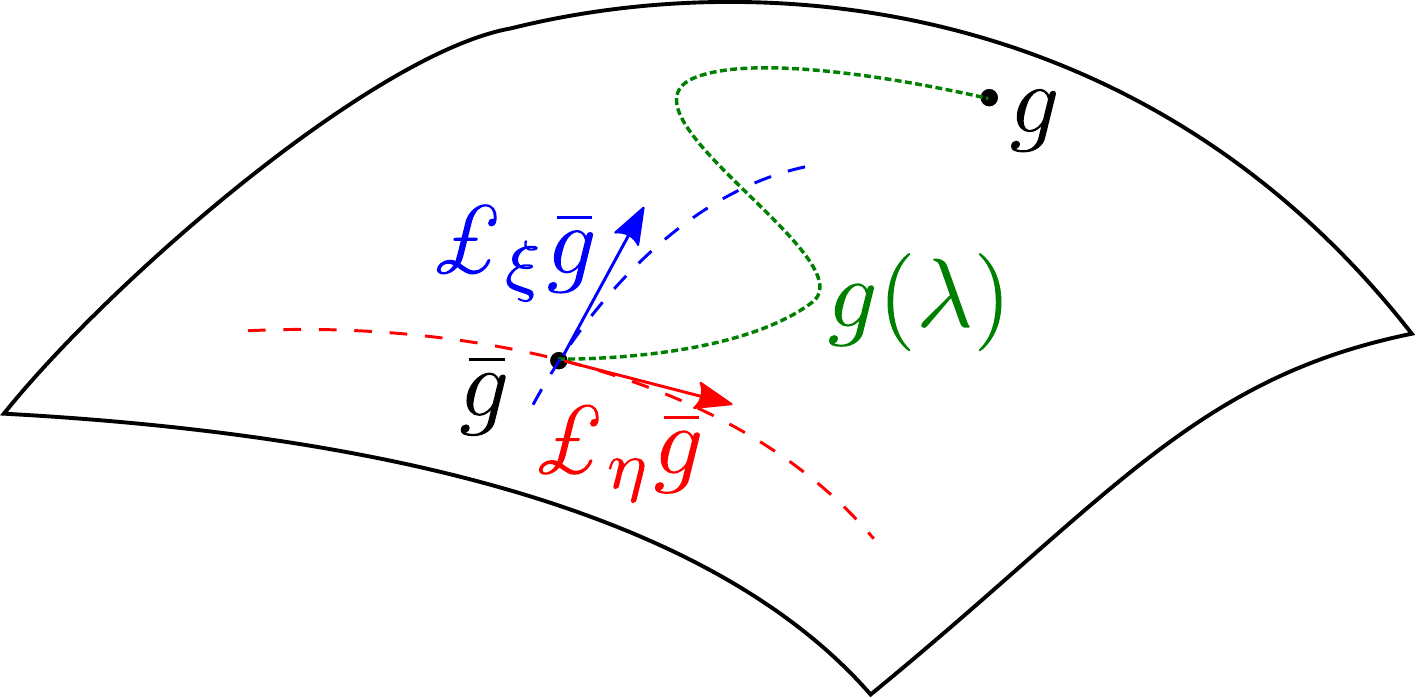}
    \caption{A schematic picture of the set of metrics we analyze in this paper. The vector fields $\xi$ and $\eta$ are elements of a Lie algebra $\mathcal{A}$. Metrics are connected to the background metric $\bar{g}_{\mu\nu}$ by diffeomorphisms generated by $\mathcal{A}$. For any metric $g_{\mu\nu}$, there exists a smooth path $g_{\mu\nu}(\lambda)$ from $\bar{g}_{\mu\nu}$ to $g_{\mu\nu}$.}
    \label{fig:configuration_space}
\end{figure}

Here, we will provide a guiding principle to find a Lie algebra $\mathcal{A}$ as a building block of the symmetries.
In most cases, even if the charges are integrable, the diffeomorphisms generated by the Lie algebra correspond to a gauge freedom.
For example, consider an algebra formed by vector fields with support in a finite spatial region in $\Sigma$ far away from the boundary $\partial\Sigma$.
In this case, although the charges are trivially integrable, all Poisson brackets of charges vanish, implying that the diffeomorphisms generated by the algebra is a gauge freedom since metrics connected by them are physically indistinguishable.
As we have already mentioned in Sec.~\ref{sec:conventional}, in the conventional approach, such a failure often happens.
In order to find a non-trivial algebra of charges, it is required that
\begin{align}
    \label{delH}
    \delta_{\eta}H[\xi] \neq 0,
\end{align}
or equivalently,  $\{H[\xi],H[\eta]\}\neq 0$ holds for some vector fields $\eta,\xi$ in the algebra.
Here, $\delta_\eta$ denotes a variation of metric such that $\delta_\eta g_{\mu\nu}=\pounds_\eta g_{\mu\nu}$.
From Eq.~\eqref{eq_definition_S}, the left hand side of Eq.~\eqref{delH} can be recast into
\begin{align}
    \label{non_triviality}
    \int _{\partial\Sigma}(\dd^{d-2}x)_{\mu\nu}S^{\mu\nu}(g,\pounds_\eta g,\pounds_\xi g)                                                                                   & \neq 0\nn                                                                 \\
    \iff \int_{\partial\Sigma}(\dd x^{d-2})_{\mu\nu}\sqrt{-g}\left[(2\nabla^{\alpha}\eta^{\mu}\nabla_{\alpha}\xi^{\nu}-\nabla_{\alpha}\eta^{\alpha}\nabla^{\mu}\xi^{\nu}\right. & \left.+\nabla_{\alpha}\xi^{\alpha}\nabla^{\mu}\eta^{\nu}) \nn\right.             \\
                                                                                                                                                                            & \left.-C_{\alpha\beta}^{\ \ \ \mu\nu}\xi^{\alpha}\eta^{\beta}\right] \neq 0,
\end{align}
where $C_{\alpha\beta\mu\nu}$ is the Weyl tensor. 

The differomorphism associated with the algebra cannot be gauged away as long as there exist $\eta,\xi$ and $g_{\mu\nu}$ satisfying Eq.~\eqref{non_triviality}.
In particular, here we adopt the following sufficient condition for Eq.~\eqref{non_triviality}:
\begin{align}
    \int _{\partial \Sigma}(\dd^{d-2}x)_{\mu\nu}S^{\mu\nu}(\bar{g},\pounds_{\eta} \bar{g},\pounds_{\xi} \bar{g}) \neq 0\label{eq_non-triviality_background}
\end{align}
as the guiding principle.
Of course, the integrability condition in Eq.~\eqref{condition1} must be satisfied. It can be recast into
\begin{align}
    0=\int_{\partial \Sigma}\left(\dd^{d-2}x\right)_{\mu\nu}\xi^{[\mu}\partial_\alpha S^{\nu]\alpha}\left(g,\pounds_\eta g,\pounds_\chi g\right), \quad \forall \xi,\eta,\chi \in \mathcal{A}\label{eq_integrability}
\end{align}
where we have used Eq.~\eqref{set2}.

For a given background metric $\bar{g}_{\mu\nu}$, it takes much efforts to find an appropriate Lie algebra $\mathcal{A}$ so that Eqs.~\eqref{eq_integrability} and \eqref{non_triviality} hold.
It corresponds to the difficulties to find an appropriate asymptotic behavior of $\delta g_{\mu\nu}$ by trials and errors in the conventional approach.
We propose the following six steps as a practical and useful way to overcome these difficulties:
\begin{enumerate}[Step~1]
    \item Fix a background metric $\bar{g}_{\mu\nu}$ of interest.
    \item
          For a fixed background metric $\bar{g}_{\mu\nu}$ of interest, find two vector fields $V_1$ and $V_2$ satisfying Eq.~\eqref{eq_non-triviality_background}.
          These are the candidates generating non-trivial diffeomorphisms whose charges are integrable.
    \item
          Introduce the minimal Lie algebra $\mathcal{A}$ including $V_1,V_2$ by calculating their commutators.
          Check whether the integrability condition at the background metric, i.e.,
          \begin{align}
              \int_{\partial \Sigma}\left(\dd^{d-2}x\right)_{\mu\nu}\xi^{[\mu}\partial_\alpha S^{\nu]\alpha}\left(\bar{g},\pounds_\eta \bar{g},\pounds_\chi \bar{g}\right)=0, \quad\forall \xi,\eta,\chi \in \mathcal{A}\label{eq_integrability_background}
          \end{align}
          is satisfied for the algebra $\mathcal{A}$ as a necessary condition for Eq.~\eqref{eq_integrability}.
          If it holds, go to the next step. Otherwise, go back to Step~2.
    \item Construct a set of metrics $g_{\mu\nu}$ which are connected to the background metric $\bar{g}_{\mu\nu}$ via differomorphisms generated by $\mathcal{A}$.
    \item Check the integrability condition in Eq.~\eqref{condition1}.
          If it is satisfied, then go to the following step.
          If not, go back to Step~2.
    \item Calculate the charges by using Eq.~\eqref{eq_charge_int_along_path}. Here, we fix the reference metric as the background metric: $g^{(0)}_{\mu\nu}=\bar{g}_{\mu\nu}$.
\end{enumerate}
An advantage of the above algorithmic protocol is the fact that Steps~2 and 3 can be done by using only the background metric $\bar{g}_{\mu\nu}$.
In particular, it should be noted that no trials and errors are required to calculate the left hand side of Eq.~\eqref{eq_non-triviality_background}. 
Furthermore, the diffeomorphism generated by $\mathcal{A}$ cannot be gauged away since the corresponding charge algebra has non-vanishing Poisson bracket by construction.
This may significantly reduce the efforts involved in finding an appropriate algebra and asymptotic behaviors of the metric in the conventional approach.
In other words, Eq.~\eqref{eq_non-triviality_background} is the guiding principle to find a non-trivial charge algebra.
Such a guiding principle does not exist in the conventional approach.
A flow chart of our approach is shown in Fig.~\ref{flowchart}.
\begin{figure}[H]
    \centering
    \includegraphics[width=18.5cm, keepaspectratio]{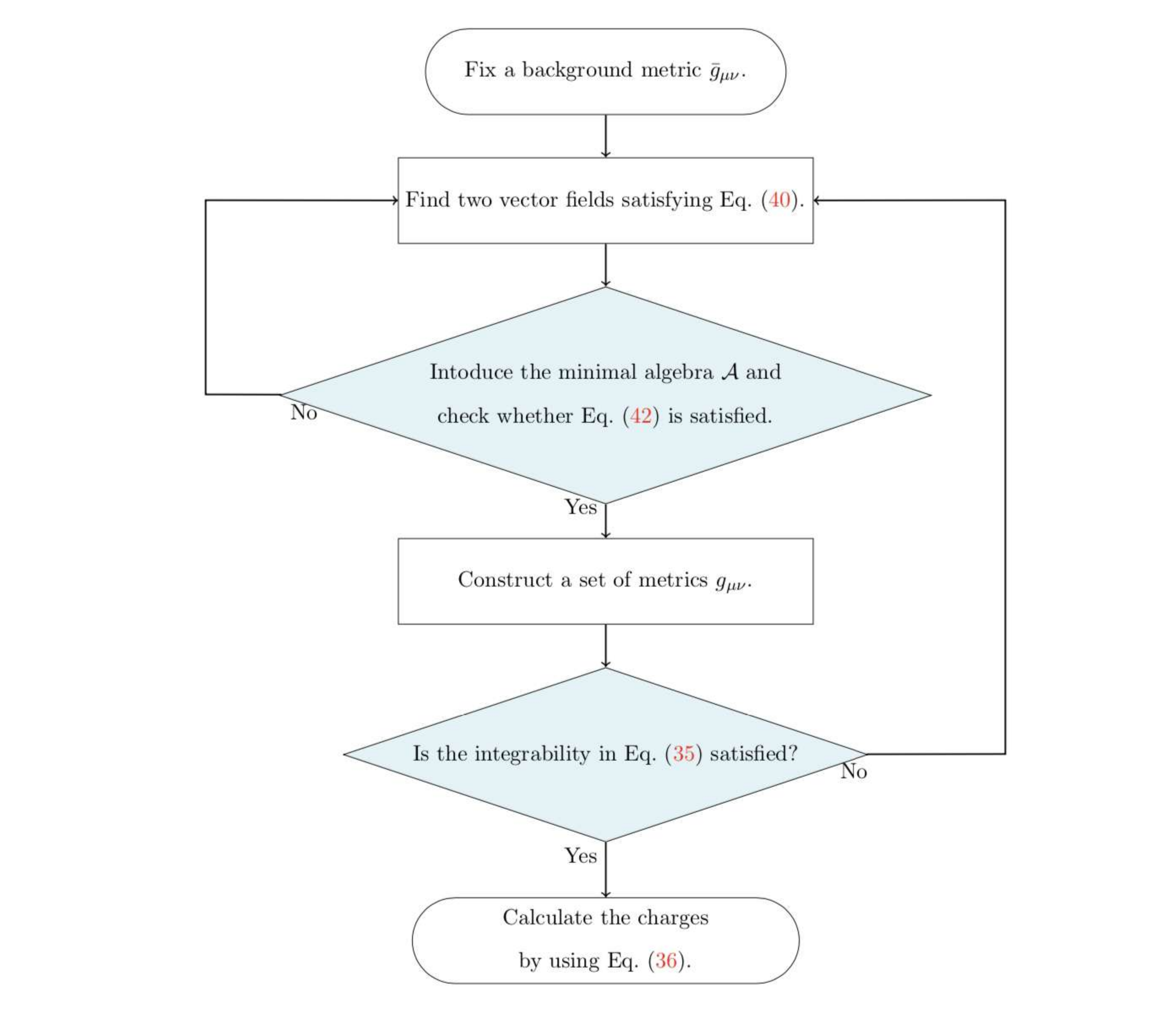}
    \caption{Flow chart of our approach.}
    \label{flowchart}
\end{figure}
Indeed our approach is quite powerful.
In the following section, we will apply our approach to the Rinlder spacetime as a demonstration.
We successfully find a new class of asymptotic symmetries on the Rindler horizon.

\section{Asymptotic symmetries on Rindler horizon}\label{sec:asymptotic_sym_in_Rindler}
In this section, we demonstrate our approach in the case where the background metric is $(1+3)$-dimensional Rindler spacetime.
In particular, we will investigate asymptotic symmetries on the Rindler horizon.

\vspace{\baselineskip}
\uline{Step1  : Fix a background metric $\bar{g}_{\mu\nu}$.} \\
Here, the background metric is fixed to be the Rindler metric given by
\begin{align}
    \dd\bar{s}^{2} = -\kappa^{2}\rho^{2}\dd\tau^{2} + \dd\rho^{2} + \dd y^{2}+\dd z^{2},
\end{align}
where $-\infty<\tau<\infty$, $0<\rho<\infty$, $-\infty<y<\infty$, $-\infty<z<\infty$ and $\kappa > 0$ is a constant.
The Rindler horizon is located at $\rho = 0$.

\vspace{\baselineskip}
\uline{Step 2 : Select two vector fields $V_1$ and $V_2$ satisfying Eq.~\eqref{eq_non-triviality_background}.}\\
Since we are interested in asymptotic symmetries in Rindler spacetime, we will analyze diffeomorphisms which map a point in the Rindler spacetime into itself. Let $\xi$ be the Lie algebra of such a diffeomorphism.
Through an infinitesimal diffeomorphism generated by $\xi$, a point $x$ of the spacetime is mapped into
\begin{align}
    x^\mu\mapsto x^\mu +\epsilon\xi^\mu +\mathcal{O}(\epsilon^2)\quad (\epsilon \to 0).
\end{align}
Since the Rindler horizon is located at $\rho=0$ in our coordinate system, unless the $\rho$-component of the vector field $\xi$ vanishes in the limit $\rho\to 0$, a point inside (resp. outside) the Rindler horizon can be mapped to the outside (resp. inside). Therefore, we require that the vector field $\xi$
has the following asymptotic behavior
\begin{align}
    \xi^\tau=\mathcal{O}(1),\quad \xi^\rho=\mathcal{O}(\rho),\quad\xi^y=\mathcal{O}(1),\quad \xi^z=\mathcal{O}(1)\quad (\rho \to 0 )
\end{align}
near the Rindler horizon.
We assume that the vector fields have supports in a finite region near the Rindler horizon.
In general, the components of the vector fields $V_{1}$ and $V_{2}$ can be written for $\rho \to 0$ as
\begin{align}
    V_{1} & =(X^{\tau}(\tau,y,z)+{\mathcal O}(\rho),X^{\rho}(\tau,y,x)\rho+{\mathcal O}(\rho^{2}),X^{A}(\tau,y,z)+{\mathcal O}(\rho)),\nonumber \\ V_{2}&=(Y^{\tau}(\tau,y,z)+{\mathcal O}(\rho),Y^{\rho}(\tau,y,z)\rho+{\mathcal O}(\rho^{2}),Y^{A}(\tau,y,z)+{\mathcal O}(\rho)),
\end{align}
where $A$ runs over $y$ and $z$.
Equation~\eqref{eq_non-triviality_background} is evaluated as
\begin{align}
     & \int_{\partial\Sigma}(\dd^{d-2}x)_{\mu\nu}S^{\mu\nu}(\bar{g},\pounds_{V_2}\bar{g},\pounds_{V_1}\bar{g})\nonumber \\
     & =\frac{1}{8\pi G\kappa}\int_{\mathbb R^{2}}	 \left[(2\kappa^{2}Y^{\tau} + \partial_{\tau}Y^{\rho})\partial_{A}X^{A} + \partial_{\tau}X^{\rho}\partial_{\tau}Y^{t} - (X \leftrightarrow Y)\right] \dd y\dd z
    \label{eq_non-triviality_Rindler},
\end{align}
where we took the limit $\rho\to 0$ in the second line since the Rindler horizon is located at $\rho= 0$. 
From this formula, we can identify several candidates for vector fields which yield a non-trivial charge algebra.

As a known example, consider the case where $X^\rho=Y^\rho=0$.
If $Y^\tau$ and $\partial_AX^A$ do not vanish, the corresponding Poisson bracket does not vanish. In this case, the vector fields $V_1$ and $V_2$ correspond to  a special class of diffeormorhisms called superrotation and supertranslation, respectively, which are shown to be integrable on the Rindler horizon in Ref.~\cite{Hotta_2016}. See Appendix~\ref{sec:app_integrability} for detailed calculations of the charges. 

Another interesting candidate, which we will investigate in detail here, is the case where $X^\rho=X^A=0$ and $Y^\tau=Y^A=0$.
If $\int \dd y \dd z\partial_\tau X^\tau \partial_\tau Y^\rho\neq 0$, the Poisson bracket does not vanish.
The vector field $V_1=(X^\tau + \mathcal{O}(\rho),0,0,0)$ generates a class of dilatation transformation in time direction since $\partial_\tau X^\tau\neq 0$ must hold.
On the other hand, the vector field $V_2=(0,\rho Y^\rho+\mathcal{O}(\rho^{2}),0,0)$
generates a dilatation in $\rho$ direction.
We term these two transformations superdilatations since the generators depend on the position in spacetime in general.

As a particular case, we will analyze the charges corresponding to two vector fields as $\rho \to 0$
\begin{align}
    V_1&=(\tau T_{1}(y,z)+\mathcal{O}(\rho^{2}),\mathcal{O}(\rho^{2}),\mathcal{O}(\rho^{2}),\mathcal{O}(\rho^{2})), \nn\\
    V_2 &=(\mathcal{O}(\rho^{2}), \tau\rho T_{2}(y,z)+\mathcal{O}(\rho^2),\mathcal{O}(\rho^{2}),\mathcal{O}(\rho^{2}))
\label{eq_vector_fields_st_sd},
\end{align}
where $T_{1}$ and $T_{2}$ are arbitrary functions of $y,z$.

\vspace{\baselineskip}
\uline{Step 3: Construct the Lie algebra including $V_1$ and $V_2$ and check the integrability at the background metric $\bar{g}_{\mu\nu}$.}\\
Since the vector fields in Eq.~\eqref{eq_vector_fields_st_sd} satisfy
\begin{align}
    [V_1,V_2]=V_3,
\end{align}
where
\begin{align}
    V_3 =(\mathcal{O}(\rho^{2}), \tau\rho T_{3}(y,z)+\mathcal{O}(\rho^2),\mathcal{O}(\rho^{2}),\mathcal{O}(\rho^{2})),\quad T_3(y,z)\coloneqq T_1(y,z)T_2(y,z),
\end{align}
the algebra $\mathcal{A}$ defined by
\begin{align}
    &\mathcal{A}\nonumber \\
    &\coloneqq \left\{V=(\tau T_{1}(y,z)+\mathcal{O}(\rho^{2}), \tau\rho T_{2}(y,z)+\mathcal{O}(\rho^2),\mathcal{O}(\rho^{2}),\mathcal{O}(\rho^{2}))\mid T_1,T_2 \text{ are arbitrary functions of $y,z$}\right\}
\end{align}
forms a closed algebra.
A straightforward calculation shows that Eq.~\eqref{eq_integrability_background}, i.e., the integrability condition at the background metric, is satisfied.

\vspace{\baselineskip}
\uline{ Step 4: Calculate the set of metrics.}\\
Since we investigate the asymptotic symmetries near the Rindler horizon, let us identify the asymptotic behavior of all the diffeomorphisms $\phi^\mu(x)$ generated by the Lie algebra $\mathcal{A}$. 

We here first calculate the asymptotic behavior of the diffeomorphisms in the form of 
$\phi_\xi^\mu(x)\coloneqq \exp[\xi](x^\mu)$ for $\xi\in\mathcal{A}$, where $\exp[\ ]$ is an exponential map. 

Introducing a real parameter $\lambda$ and calculating the integral curve $\varphi^{\mu}_{\lambda}(x)\coloneqq \exp[\lambda \xi](x^\mu)$  of the vector field $\xi$, the diffeomorphism $\phi_{\xi}^\mu(x)$ is given by $\phi_{\xi}^\mu(x)=\varphi^\mu_{\xi;\lambda=1}(x)$. 
The integral curve is the solution of the following differential equation:
\begin{align}
    \frac{\dd}{\dd\lambda}\varphi_{\xi;\lambda}^\mu(x) =\xi^\mu (\varphi(x)),\quad \varphi_{\xi;0}^\mu(x)=x^\mu.\label{eq_integral_curve_ODE}.
\end{align}
Any vector field $\xi$ of the algebra $\mathcal{A}$ can be decomposed into two parts:
\begin{align}
    \xi^\mu(x) &=\Xi^\mu(x)+h^\mu(x),\\
    \Xi^\mu(x)&\coloneqq (\tau F_{1}(y,z), \tau\rho F_{2}(y,z),0,0),\\
    h^\mu(x)&\coloneqq 
    (\mathcal{O}(\rho^{2}),\mathcal{O}(\rho^2),\mathcal{O}(\rho^{2}),\mathcal{O}(\rho^{2})) \quad (\rho \to 0),
\end{align}
where $F_1$ and $F_2$ are arbitrary functions of $(y,z)$. When $\xi=\Xi$, the solution of the differential equation is straightforwardly calculated as
\begin{align}
    \varphi^\mu_{\Xi;\lambda}(x)=\left(\tau e^{F_{1}(y,z)\lambda}, \rho\exp\left(\frac{F_{2}(y,z)}{F_{1}(y,z)}\tau\left(e^{F_{1}(y,z)\lambda}-1\right)\right), y,z\right).
\end{align}
In Appendix~\ref{Flow_proof}, it is proven that
\begin{align}
    \varphi_{\xi;\lambda}^\mu(x)= \varphi_{\Xi;\lambda}(x)+(\mathcal{O}(\rho^{2}),\mathcal{O}(\rho^2),\mathcal{O}(\rho^{2}),\mathcal{O}(\rho^{2})) \quad (\rho \to 0).
\end{align}
This is the asymptotic behavior of the integral curve. 
Thus, the asymptotic behavior of the diffeomorphism $\phi_\xi^\mu(x)=\exp[\xi](x^\mu)$ is given by
\begin{align}
    \phi_\xi^\mu(x)&=\varphi_{\xi;\lambda=1}^\mu(x)\nn\\
    &=\phi_{\Xi}^\mu(x)+(\mathcal{O}(\rho^{2}),\mathcal{O}(\rho^2),\mathcal{O}(\rho^{2}),\mathcal{O}(\rho^{2}))\nn\\
    &=\left(\tau e^{F_{1}(y,z)}, \rho\exp\left(\frac{F_{2}(y,z)}{F_{1}(y,z)}\tau\left(e^{F_{1}(y,z)}-1\right)\right), y,z\right)+(\mathcal{O}(\rho^{2}),\mathcal{O}(\rho^2),\mathcal{O}(\rho^{2}),\mathcal{O}(\rho^{2}))\label{eq:asymp_single}
\end{align}
as $\rho\to 0$. 

So far, we have calculated the asymptotic behavior of the diffeomorphisms in the form of $\phi_\xi^\mu(x)=\exp[\xi](x^\mu)$ for $\xi\in \mathcal{A}$. 
In general, diffeomorphisms generated by $\mathcal{A}$ and  connected to the identity transformation are given by a product of such maps, i.e.,
\begin{align}
    (\phi_{\xi^{(1)}}\circ \phi_{\xi^{(2)}}\circ\cdots \circ \phi_{\xi^{(N)}})(x)
\end{align}
for some $N$. 
Let us analyze the asymptotic behavior for $N=2$. For two vector fields
\begin{align}
    \left(\xi^{(i)}\right)^{\mu}(x) & =(\tau F_{1}^{(i)}(y,z)+\mathcal{O}(\rho^2), \tau\rho F_{2}^{(i)}(y,z)+\mathcal{O}(\rho^2),\mathcal{O}(\rho^2),\mathcal{O}(\rho^2)),\quad i=1,2,
\end{align}
as $\rho\to 0$, 
Eq.\eqref{eq:asymp_single} implies that
\begin{align}
    &(\phi_{\xi^{(1)}}\circ \phi_{\xi^{(2)}})^\mu(x)\nn\\
    &=\left(\tau e^{\tilde{F}_{1}(y,z)}, \rho\exp\left(\frac{\tilde{F}_{2}(y,z)}{\tilde{F}_{1}(y,z)}\tau\left(e^{\tilde{F}_{1}(y,z)}-1\right)\right),y,z\right)+(\mathcal{O}(\rho^{2}),\mathcal{O}(\rho^2),\mathcal{O}(\rho^{2}),\mathcal{O}(\rho^{2})),
\end{align}
where we have defined
\begin{align}
    \tilde{F}_1(y,z)&\coloneqq F_1^{(1)}(y,z)+F_1^{(2)}(y,z),\nn\\ \tilde{F}_2(y,z)&\coloneqq \tilde{F}_1(y,z)\left(\frac{F_{2}^{(2)}(y,z)}{F_{1}^{(2)}(y,z)}\left(e^{F_{1}^{(2)}(y,z)}-1\right) +\frac{F_{2}^{(1)}(y,z)}{F_{1}^{(1)}(y,z)} e^{F^{(2)}_1(y,z)}\left(e^{F_{1}^{(1)}(y,z)}-1\right)\right).
\end{align}
Repeating the same argument, it is shown that the asymptotic behavior of a general diffeomorphism $\chi_{(F_1,F_2)}$ is characterized by two real functions $F_1$ and $F_2$ of $(y,z)$ as
\begin{align}
   \chi_{(F_1,F_2)}^\mu(x) =\left(\tau e^{F_{1}(y,z)}, \rho\exp\left(\frac{F_{2}(y,z)}{F_{1}(y,z)}\tau\left(e^{F_{1}(y,z)}-1\right)\right), y,z\right)+(\mathcal{O}(\rho^{2}),\mathcal{O}(\rho^2),\mathcal{O}(\rho^{2}),\mathcal{O}(\rho^{2}))
\end{align}
for $\rho \to 0$. 

Thus, the asymptotic behavior of the components of the metrics in question is characterized by arbitrary functions $F_1$ and $F_2$ of $(y,z)$ as
\begin{align}
    \left(g_{\mu\nu}^{(F_1,F_2)}(x)\right) & \coloneqq \left(\frac{\partial \chi_{(F_{1},F_{2})}^{\alpha}}{\partial x^{\mu}}\frac{\partial \chi_{(F_{1},F_{2})}^{\beta}}{\partial x^{\nu}}\bar{g}_{\alpha\beta}(\chi_{(F_{1},F_{2})}(x))\right)\nonumber \\
                               & =
    \begin{pmatrix}
        J_{11}\rho^{2} & J_{12}\rho & J_{1y}\rho^{2} & J_{1z}\rho^{2} \\
        J_{12}\rho     & J_{22}     & J_{2y}\rho     & J_{2z}\rho     \\
        J_{1y}\rho^{2} & J_{2y}\rho & 1              & 0              \\
        J_{1z}\rho^{2} & J_{2z}\rho & 0              & 1
    \end{pmatrix} + (\text{higher order term}),\label{eq:asym_metric}
\end{align}
where we have defined
\begin{align}
    J_{11} (\tau,y,z)&\coloneqq  e^{2f(y,z)\tau}\left(-\kappa^{2}e^{2F_{1}(y,z)}+f^{2}(y,z)\right),\nn \\
    J_{12}(\tau,y,z) & \coloneqq f(y,z)e^{2f(y,z)\tau},\ J_{1A}(\tau,y,z)\coloneqq \tau e^{2f(y,z)\tau}(-\kappa^{2}\partial_{A}F_{1}(y,z)e^{2F_{1}(y,z)} + f(y,z)\partial_{A}f(y,z)) \nn ,\\
    J_{22}(\tau,y,z) & \coloneqq e^{2f(y,z)\tau},\ J_{2A}(\tau,y,z) \coloneqq \tau\partial_{A}f(y,z)e^{2f(y,z)\tau},
\end{align}
and
\begin{align}
    f (y,z)\coloneqq  \frac{F_{2}(y,z)}{F_{1}(y,z)}\left(e^{F_{1}(y,z)}-1\right).
\end{align}
As explicit calculations show, it turns out that the second term in Eq.~\eqref{eq:asym_metric} does not affect the integrability condition nor the expression of the charges. 

\vspace{\baselineskip}
\uline{Step 5: Check the integrability condition.}\\
A straightforward but lengthy calculation shows that the integrand of Eq.~\eqref{eq_integrability} is $\mathcal{O}(\rho)$ as $\rho\to 0$ for any metric given in Eq.~\eqref{eq:asym_metric}. Therefore, the integrability condition is satisfied. 

\vspace{\baselineskip}
\uline{Step 6: Calculate the charges.}\\
To calculate the charges for $V_{1}, V_{2}$ defined in Eq.~\eqref{eq_vector_fields_st_sd}, we need $Q^{\tau\rho}, \Theta^{\tau}$ and $\Theta^{\rho}$ in Eq.~\eqref{eq_charge_int_along_path}. 
Since the integrability condition is satisfied, the parametrization of metric in Eq.~\eqref{eq_charge_int_along_path} can be taken arbitrarily.
In order to calculate the charges at metric $g_{\mu\nu}^{(F_1,F_2)}(x)$ given in Eq.~\eqref{eq:asym_metric}, 
we here adopt following: 
\begin{align}
    g_{\mu\nu}(x;\lambda) & =\frac{\partial \chi_{(\lambda F_{1},\lambda F_{2})}^{\alpha}}{\partial x^{\mu}}\frac{\partial \chi_{(\lambda F_{1},\lambda F_{2})}^{\beta}}{\partial x^{\nu}}\bar{g}_{\alpha\beta}(\chi_{(\lambda F_{1},\lambda F_{2})}(x)).
\end{align}
For $\lambda=1$, $
    g_{\mu\nu}(x;\lambda=1)=g_{\mu\nu}^{(F_1,F_2)}(x)$,
while for $\lambda =0$, $
    \left(g_{\mu\nu}(x;\lambda=0)\right)=\left(\bar{g}_{\mu\nu}(x)\right)$ up to the higher order terms in Eq.~\eqref{eq:asym_metric}, which does not affect the charges, shown as follows: 
From Eq.~\eqref{eq:Komar}, we get
\begin{align}
	Q^{\tau\rho}\left[V_{1}\right]\biggl|_{\left(g_{\mu\nu}(x;\lambda)\right)} & = \frac{ T_{1}}{8\pi G \kappa}e^{-\lambda F_{1}}\left(\kappa^{2}e^{2\lambda F_{1}}\tau + \frac{f}{2}\right)+\mathcal{O}(\rho) \\
	Q^{\tau\rho}\left[V_{2}\right] \biggl|_{\left(g_{\mu\nu}(x;\lambda)\right)}& = \frac{T_{2}}{16\pi G \kappa}e^{-\lambda F_{1}}+\mathcal{O}(\rho)
\end{align}
as $\rho\to 0$. On the other hand,  from Eq.~\eqref{eq:presymplectic_potential}, we have
\begin{align}
	\Theta^{\tau} & =\mathcal{O}(\rho)\\
	\Theta^{\rho} & =-\frac{\kappa }{8\pi G}\partial_{\lambda} (e^{\lambda F_{1}})+\mathcal{O}(\rho)
\end{align}
as $\rho \to 0$. 
Thus, the second term in Eq.~\eqref{eq:asym_metric} does not contribute to the expression of the charges. 

From Eq.~\eqref{eq_charge_int_along_path}, the charges are evaluated as
\begin{align}
    H[V_{1}] & = \frac{1}{16\pi G\kappa}\int \dd y\dd z\  T_{1}(y,z)\frac{F_{2}(y,z)}{F_{1}(y,z)}\left(1-e^{-F_{1}(y,z)}\right), \label{eq_charge_V_1_Rindler} \\
    H[V_{2}] & =\frac{1}{16\pi G\kappa}\int \dd y\dd z\  T_{2}(y,z)\left(e^{-F_{1}(y,z)}-1\right).
    \label{eq_charge_V_2_Rindler}
\end{align}
where the reference of the charges are chosen so that they vanish at the background metric, which corresponds to the case where $F_1=F_2=0$.
The transformation generated by the vector fields $V_1$ and $V_2$ is an example of superdilatation.
Since the Rindler horizon can be interpreted as the horizon of a Schwarzschild black hole in the limit of infinite black hole mass, it would be interesting to investigate a similar asymptotic symmetry on the latter one.
To the authors' knowledge, the algebra of charges corresponding to the supardilatation on the horizon has not been investigated neither in the Rindler spacetime nor in the Schwarzschild spacetime in prior researches.
\section{summary}\label{sec:summary}
In this paper, we have proposed a useful approach to construct integrable charges which form a non-trivial algebra in general spacetime.
Our approach using the guiding principle in Eq.~\eqref{eq_non-triviality_background} may significantly reduce the effort involved in finding proper asymptotic conditions by trials and errors in the conventional approach.
In particular, a key idea of our approach is to use Eq.~\eqref{eq_non-triviality_background} to find an algebra of symmetries with a non-vanishing Poisson bracket at the background metric $\bar{g}_{\mu\nu}$. The metrics connected to the background metric through a diffeomorphism generated by the Lie algebra $\mathcal{A}$ satisfying Eq.~\eqref{eq_non-triviality_background} can be physically distinguished from each other since the Poisson brackets do not vanish.

In our analysis, we have investigated a set of metrics which are connected to a fixed background metric by diffeomorphisms generated by a Lie algebra of vector fields.
Since all the metrics are diffeomorphic to the background metric, it is possible to investigate the properties of the asymptotic symmetries of the background spacetime.
The set in our approach is different from that in the conventional approach, where the set of metrics are defined by their asymptotic behaviors.

As an explicit example, we have analyzed the asymptotic symmetries on the Rindler horizon in $(1+3)$-dimensional Rindler spacetime.
Equation~\eqref{eq_non-triviality_Rindler} is the general result of Eq.~\eqref{eq_non-triviality_background} for arbitrary vector fields evaluated at the Rindler horizon.
From this formula, we can read out candidates of transformations which yields a non-trivial charge algebra.
It is shown that the supertranslation and superrotation on the Rindler horizon can be found in our approach, which is known to be integrable \cite{Hotta_2016}.
In addition, we have found a new class of symmetries, which generates position dependent dilatations in  time and in the direction perpendicular to the horizon.
We have termed such a transformation superdilatation.
For a concrete example of superdilatation algebra, we have shown that the charges are integrable. The explicit expressions of the charges are given in Eqs.~\eqref{eq_charge_V_1_Rindler} and \eqref{eq_charge_V_2_Rindler}.
Of course, our analysis here in $(1+3)$-dimensional Rindler spacetime can be directly extended to $(1+D)$-dimensional case with any $D\geq 2$. 
It remains an open problem whether there are such dilatation-like asymptotic symmetries in other setups. It will also be interesting to investigate whether known results can be reproduced with our approach, e.g., a class of dilatations at null infinity of asymptotic flat spacetime \cite{Haco2017}. 

So far, we have started with two vector fields $V_1$ and $V_2$ satisfying Eq.~\eqref{eq_integrability_background} and constructed a minimal Lie algebra $\mathcal{A}$ spanned by the vector fields and their commutators.
This approach enables us to find building blocks of the asymptotic symmetries.
To proceed the classification of the symmetry in general relativity, it will be quite interesting to investigate how the charge algebra changes by adding other elements to $\mathcal{A}$.
It is also interesting to derive a condition under which Eq.~\eqref{non_triviality} holds at a particular metric $g_{\mu\nu}$ but not at the background metric $\bar{g}_{\mu\nu}$.

Although we have used our approach to analyze the asymptotic symmetries on the Rindler horizon, it is applicable to an arbitrary spacetime. 
For background spacetimes without symmetry, it may turn out that the left hand side of Eq.~\eqref{eq_non-triviality_background} vanishes, suggesting that there is no asymptotic symmetry. 
We expect that our approach will be helpful to investigate other important spacetimes, such as black holes, the de Sitter spacetime and the anti-de Sitter spacetime.

\begin{acknowledgments}
    The authors thank Ursula Carow-Watamura, Hiroyuki Kitamoto, Kohei Miura, Kengo Shimada, Naoki Watamura, Satoshi Watamura, Masaki Yamada and Kazuya Yonekura for useful discussions.
    This research was partially supported by JSPS KAKENHI Grants No. JP18J20057 (K.Y.), No. JP19K03838 (M.H.) and 21H05188(M.H.), and by Graduate Program on Physics for the Universe of Tohoku University (T.T. and K.Y.).
\end{acknowledgments}
\appendix
\section{An integral curve of vector field}
\label{Flow_proof}
In this appendix, we show that $\mathcal{O}(\rho^{2})$ terms in a vector field result in $\mathcal{O}(\rho^{2})$ terms in its integral curve. 

Let us define a vector field
\begin{align}
\xi^{\mu}(x) \coloneqq \Xi^{\mu}(x) + h^{\mu}(x)
\end{align}
where
\begin{align}
\Xi^{\mu}(x) &= (X^{\tau}(\tau,y,z), X^{\rho}(\tau,y,z)\rho, X^{y}(\tau,y,z), X^{z}(\tau,y,z) ), \label{X}\\
h^{\mu}(x) &= (\mo{\rho}{2}, \mo{\rho}{2}, \mo{\rho}{2}, \mo{\rho}{2})\ \ \ (\rho \to 0).
\label{h}
\end{align}
The integral curve of $\xi^{\mu}$ is defined as
\begin{align}
\varphi_{\xi;\lambda}^{\mu}(x) \coloneqq \exp[\lambda \xi]x^{\mu}=\sum_{n=0}^{\infty}\frac{\lambda^n}{n!}\xi^{n}x^{\mu}.
\end{align}
Where the action of $\xi^n$ on a function of $x^{\mu}$ is recursively defined as
\begin{align}
    \xi^{n}f(x) &= \xi^{n-1}\xi^{\mu}(x)\partial_{\mu}f(x) \qquad (n=1,2,3,\cdots), \\
    \xi^{0}f(x) &= f(x).
\end{align}
Defining
\begin{align}
\varphi_{\xi;\lambda,n}^\mu( x) \coloneqq \frac{\lambda^{n}}{n!}\xi^{n}x^{\mu},
\end{align}
we will show the following proposition:
\begin{proposition}
$\forall n \in \mathbb{Z}_{\geq 0}$,
\begin{align}
\varphi_{\xi;\lambda,n}^\mu( x) =\frac{\lambda^{n}}{n!} \Xi^{n}x^{\mu} + \epsilon_{n}^{\mu}(\lambda,x)
\end{align}
where the asymptotic behavior of the first term is $(\mo{1}{}, \mo{\rho}{}, \mo{1}{}, \mo{1}{})$
and that of $\epsilon_{n}^{\mu}(\lambda,x)$ is $(\mo{\rho}{2}, \mo{\rho}{2}, \mo{\rho}{2}, \mo{\rho}{2})$ as $\rho \to 0$.
\end{proposition}
\begin{proof}
We give a proof by induction on $n$.
Since $\varphi_{\xi;\lambda,0}^{\mu}(x) = x^{\mu}$, the proposition is clearly satisfied for $n=0$.
Assuming the proposition is satisfied for some integer $k \geq 0$, we have
\begin{align}
\varphi_{\xi;\lambda,k+1}^\mu(x ) &= \frac{\lambda}{k+1}\xi \varphi_{\xi;\lambda,k}^\mu(x)  \nn \\
&=\frac{\lambda}{k+1}(\Xi+ h)\left(\frac{\lambda^{k}}{k!}\Xi^{k}x^{\mu} + \epsilon_{k}^{\mu}(x)\right) \nn \\
&=\frac{\lambda^{k+1}}{(k+1)!}\left(
\Xi^{k+1}x^{\mu} + h^{\alpha}\partial_{\alpha}(\Xi^{k}x^{\mu}) \right)+\frac{\lambda}{k+1} \left(\Xi^{\alpha}\partial_{\alpha}\epsilon_{k}^{\mu}(x) + h^{\alpha}\partial_{\alpha}\epsilon_{k}^{\mu}(x)
\right).
\label{k+1}
\end{align}
By Eqs.~\eqref{X}, \eqref{h}, and the assumption for $k$, we have for each term in Eq.~\eqref{k+1}:
\begin{align}
\frac{\lambda^{k+1}}{(k+1)!}\Xi^{k+1}x^{\mu} = \frac{\lambda^{k+1}}{(k+1)!}\Xi^{\alpha}\partial_{\alpha}(\Xi^{k}x^{\mu}) &= \frac{\lambda^{k+1}}{(k+1)!}\Xi^{\alpha}\partial_{\alpha}(\mo{1}{}, \mo{\rho}{}, \mo{1}{}, \mo{1}{}) \nn \\
&=(\mo{1}{}, \mo{\rho}{}, \mo{1}{}, \mo{1}{}),
\end{align}
\begin{align}
\frac{\lambda^{k+1}}{(k+1)!}h^{\alpha}\partial_{\alpha}(\Xi^{k}x^{\mu}) &= \frac{\lambda^{k+1}}{(k+1)!}h^{\alpha}\partial_{\alpha}(\mo{1}{}, \mo{\rho}{}, \mo{1}{}, \mo{1}{}) \nn \\
&=(\mo{\rho}{2}, \mo{\rho}{2}, \mo{\rho}{2}, \mo{\rho}{2}),
\end{align}
\begin{align}
\frac{\lambda}{k+1}\Xi^{\alpha}\partial_{\alpha}\epsilon_{k}^{\mu}(x) &=\frac{\lambda}{k+1} X^{\alpha}\partial_{\alpha}(\mo{\rho}{2}, \mo{\rho}{2}, \mo{\rho}{2}, \mo{\rho}{2})\nn \\
&=(\mo{\rho}{2}, \mo{\rho}{2}, \mo{\rho}{2}, \mo{\rho}{2}),
\end{align}
\begin{align}
\frac{\lambda}{k+1}h^{\alpha}\partial_{\alpha}\epsilon_{k}^{\mu}(x) &=\frac{\lambda}{k+1} h^{\alpha}\partial_{\alpha}(\mo{\rho}{2}, \mo{\rho}{2}, \mo{\rho}{2}, \mo{\rho}{2}) \nn \\
&=(\mo{\rho}{3}, \mo{\rho}{3}, \mo{\rho}{3}, \mo{\rho}{3}).
\end{align}
Then for $n=k+1$ the proposition is also satisfied. Thus the proposition is satisfied for $\forall n \in \mathbb{Z}_{\geq 0}$.
\end{proof}
The integral curve generated by $\xi^{\mu}$ is now 
\begin{align}
\varphi_{\xi:\lambda}^{\mu}(x) &= \sum_{n=0}^{\infty}\varphi_{\xi;\lambda,n}^\mu( x) \nn \\
&=\sum_{n=0}^{\infty}\frac{\lambda^{n}}{n!}\Xi^{n}x^{\mu} + \sum_{n=0}^{\infty}\epsilon_{n}^{\mu}(\lambda,x) \nn \\
&=\exp(\lambda \Xi)x^{\mu} +  \sum_{n=0}^{\infty}\epsilon_{n}^{\mu}(\lambda, x) ,
\label{flow}
\end{align}
where $\epsilon_{n}^{\mu}(\lambda, x)$ is defined through the following recurrence relation:
\begin{align}
\epsilon_{0}^{\mu}(\lambda, x) &= 0, \\
\epsilon_{n+1}^{\mu}(\lambda,x) &=\frac{\lambda^{n+1}}{(n+1)!}h^{\alpha}\partial_{\alpha}(\xi^{n}x^{\mu})+\frac{\lambda}{n+1} \left(\xi^{\alpha}\partial_{\alpha}\epsilon_{n}^{\mu}(x) + h^{\alpha}\partial_{\alpha}\epsilon_{n}^{\mu}(x)
\right).
\end{align}
In Eq.~\eqref{flow}, the asymptotic behavior of the first term is $(\mo{1}{}, \mo{\rho}{}, \mo{1}{}, \mo{1}{})$,
while that of the second term  is $(\mo{\rho}{2}, \mo{\rho}{2}, \mo{\rho}{2}, \mo{\rho}{2})$ as $\rho\to 0$. Taking $\lambda=1$, a diffeomorphism $\phi_\xi^\mu(x)\coloneqq \varphi_{\xi;\lambda=1}^\mu(x)$ satisfies
\begin{align}
    \phi_\xi^\mu(x)=\phi_\Xi^\mu (x)+(\mo{\rho}{2}, \mo{\rho}{2}, \mo{\rho}{2}, \mo{\rho}{2})
\end{align}
as $\rho\to 0$. 

\section{Supertranslations and Superrotation charges on Rindler horizon}\label{sec:app_integrability}
In this appendix, we analyze the charges corresponding to two vector fields such that as $\rho \to 0$,
\begin{align}
    U_{1} &= (W(y,z) + \mathcal{O}(\rho^{2}), \mathcal{O}(\rho^{2}),\mathcal{O}(\rho^{2}),\mathcal{O}(\rho^{2})),\\
    U_{2} &= (\mathcal{O}(\rho^{2}), \mathcal{O}(\rho^{2}),R^{y}(y,z)+\mathcal{O}(\rho^{2}),R^{z}(y,z)+\mathcal{O}(\rho^{2}))
\end{align}
where $W$ and $R^{A}~(A=y,z)$ are arbitrary functions of $y,z$. They generate a well-known algebra of supertranslation and superrotation.

Since they satisfy 
\begin{align}
    [U_{1} , U_{2}] = U_{3}
\end{align}
where
\begin{align}
    U_{3} = (W'(y,z)+\mathcal{O}(\rho^{2}),\mathcal{O}(\rho^{2}),\mathcal{O}(\rho^{2}),\mathcal{O}(\rho^{2})),\quad {W'}(y,z) \coloneqq -R^{A}(y,z)\partial_{A}W(y,z),
\end{align}
the algebra $\mathcal{B}$ defined by
\begin{align}
    \mathcal{B}
    &\coloneqq \left\{U=( W(y,z)+\mathcal{O}(\rho^{2}),  \mathcal{O}(\rho^2),R^{y}(y,z)+\mathcal{O}(\rho^{2}),R^{z}(y,z)+\mathcal{O}(\rho^{2}))\right.\nonumber \\
    &\hspace{5cm}\left.\mid W,R^{A} \text{ are arbitrary functions of $y,z$}\right\}
\end{align}
forms a closed algebra.
A straightforward calculation shows that the integrability condition at the background metric is satisfied.

Let us introduce a real parameter $\lambda$ and calculate the integral curve $\varsigma_{\lambda}^{\mu}(x) \coloneqq \exp[\lambda \eta](x^{\mu})$ for $ \eta \in \mathcal{B}$, which satisfies the following differential equation:
\begin{align}
\frac{\dd}{\dd\lambda}\varsigma^{\mu}_{\eta;\lambda} = \eta^{\mu}(\varsigma(x)).
\end{align}
Any vector field $\eta$ of the algebra $\mathcal{B}$ can be decomposed into
\begin{align}
    \eta^{\mu}(x) &= H^{\mu}(x) + h^{\mu}(x),\\
    H^{\mu}(x) &\coloneqq (P(y,z), 0, G^{y}(y,z), G^{z}(y,z)), \\
    h^{\mu}(x) &= (\mathcal{O}(\rho^{2}),\mathcal{O}(\rho^{2}),\mathcal{O}(\rho^{2}),\mathcal{O}(\rho^{2})) \quad (\rho \to 0),
\end{align}
where $P$ and $G$ are arbitrary functions of $(y,z)$. 
As we have shown in Appendix~\ref{Flow_proof}, the asymptotic behavior of the solution of the differential equation is given by
\begin{align}
    \varsigma^{\mu}_{\eta;\lambda}=\varsigma^{\mu}_{H;\lambda}+  (\mathcal{O}(\rho^{2}),\mathcal{O}(\rho^{2}),\mathcal{O}(\rho^{2}),\mathcal{O}(\rho^{2}))
\end{align}
as $\rho\to 0$. 

Let us analyze the case where $\eta=H$. Note that $G^{y}$ and $G^z$ are functions of $(y,z)$. In addition, the initial condition $\varsigma^A_{H;\lambda=0}$ is independent of $\tau$ and $\rho$. Thus, the $A$-component of the integral curve can generally be written as
\begin{align}
    \varsigma^{A}_{H;\lambda}(\tau,\rho,y,z)=\widetilde{G}^A(y,z;\lambda),\quad A=y,z
\end{align}
for some functions $\widetilde{G}^A$ of $y,z$ and $\lambda$. Since $P$ is a function of $(y,z)$, the $\tau$-component of the differential equation is given by
\begin{align}
    \frac{\dd}{\dd\lambda}\varsigma^{\tau}_{H;\lambda}(\tau,\rho,y,z) = P(\widetilde{G}^y(y,z;\lambda),\widetilde{G}^z(y,z;\lambda)),\quad \varsigma^{\tau}_{\eta;\lambda=0}(\tau,\rho,y,z)=\tau.
\end{align}
Its solution is written as
\begin{align}
    \varsigma^{\tau}_{\eta;\lambda}(\tau,\rho,y,z)=\tau+\widetilde{P}(y,z;\lambda),
\end{align}
where $\widetilde{P}$ is some function of $y,z$ and $\lambda$.
Therefore, in general, the asymptotic behavior of the diffeomorphism $\sigma^{\mu}_{\eta}(x) \coloneqq \exp[\eta](x^{\mu})=\varsigma^{\mu}_{\eta;\lambda=1}(x)$ is given by
\begin{align}
    \sigma^{\mu}_{\eta}(x) &=
    (\tau + \widetilde{P}(y,z), \rho,\widetilde{G}^{y}(y,z),\widetilde{G}^{z}(y,z)) +
    (\mathcal{O}(\rho^{2}),\mathcal{O}(\rho^{2}),\mathcal{O}(\rho^{2}),\mathcal{O}(\rho^{2}))
\end{align}
as $\rho \to 0$, where we have re-defined
\begin{align}
    \widetilde{P}(y,z)\coloneqq \widetilde{P}(y,z;\lambda=1),\quad \widetilde{G}^A(y,z)\coloneqq \widetilde{G}^A(y,z;\lambda=1)\qquad A=y,z.
\end{align}

As we have done at Step~4 in Sec.~\ref{sec:asymptotic_sym_in_Rindler}, it can be confirmed that the asymptotic behavior of a general diffeomorphism $\gamma^{\mu}_{(\widetilde{P},\widetilde{G})}$ is characterized by three real functions $\widetilde{P}$ and $\widetilde{G}^{A}$ of $(y,z)$ as
\begin{align}
   \gamma^{\mu}_{(\widetilde{P},\widetilde{G})}(x) = (\tau + \widetilde{P}(y,z), \rho,\widetilde{G}^{y}(y,z),\widetilde{G}^{z}(y,z)) +
    (\mathcal{O}(\rho^{2}),\mathcal{O}(\rho^{2}),\mathcal{O}(\rho^{2}),\mathcal{O}(\rho^{2}))
\end{align}
for $\rho \to 0$. Thus, the asymptotic behavior of the components of the metrics in question is characterized by arbitrary functions $\widetilde{P}$ and $\widetilde{G}^A$ of $(y,z)$ as
\begin{align}
    \left(g_{\mu\nu}^{(\widetilde{P},\widetilde{G})}(x)\right) & \coloneqq \left(\frac{\partial \gamma_{(\widetilde{P},\widetilde{G})}^{\alpha}}{\partial x^{\mu}}\frac{\partial \gamma_{(\widetilde{P},\widetilde{G})}^{\beta}}{\partial x^{\nu}}\bar{g}_{\alpha\beta}(\gamma_{(\widetilde{P},\widetilde{G})}(x))\right)\nonumber \\
                               & =
    \begin{pmatrix}
        -\kappa^{2}\rho^{2} & 0 & L_{1y}\rho^{2} & L_{1z}\rho^{2} \\
        0     & 1    & 0     & 0     \\
        L_{1y}\rho^{2} & 0 & L_{yy}              & L_{yz}              \\
        L_{1z}\rho^{2} & 0 & L_{yz}              & L_{zz}
    \end{pmatrix} + (\text{higher order term}),
\end{align}
where we have defined
\begin{align}
    L_{1A}(y,z) \coloneqq -\kappa^{2}\partial_{A}\widetilde{P}(y,z),\quad 
    L_{AB}(y,z) \coloneqq \partial_{A}\widetilde{G}^{y}(y,z)\partial_{B}\widetilde{G}^{y}(y,z) + \partial_{A}\widetilde{G}^{z}(y,z)\partial_{B}\widetilde{G}^{z}(y,z).
\end{align}
A straightforward calculation shows that the above metric satisfies the integrability condition.

Let us adopt the parametrization of metric as
\begin{align}
        g_{\mu\nu}(x;\lambda) & \coloneqq \frac{\partial \gamma_{(\lambda\widetilde{P},\lambda\widetilde{G})}^{\alpha}}{\partial x^{\mu}}\frac{\partial \gamma_{(\lambda\widetilde{P},\lambda\widetilde{G})}^{\beta}}{\partial x^{\nu}}\bar{g}_{\alpha\beta}(\gamma_{(\lambda\widetilde{P},\lambda\widetilde{G})}(x)).
\end{align}
On one hand, from Eq.~\eqref{eq:Komar}, we get
\begin{align}
	Q^{\tau\rho}[U_{1}]\biggl|_{g_{\mu\nu}(x;\lambda)} &=\lambda^{2} \frac{\sqrt{L}\kappa}{8\pi G} W + \mathcal{O}(\rho), \\
	Q^{\tau\rho}[U_{2}]\biggl|_{g_{\mu\nu}(x;\lambda)} &=- \lambda^{3}\frac{\sqrt{L}}{8\pi G \kappa}R^{A}L_{1A} + \mathcal{O}(\rho)
\end{align}
as $\rho \to 0$, where we have defined $L\coloneqq L_{yy}L_{zz}-L_{yz}^{2}$. On the other hand, from Eq.~\eqref{eq:presymplectic_potential}, we have
\begin{align}
    \Theta^{\rho} = \mathcal{O}(\rho)
\end{align}
as $\rho \to 0$.

Therefore, from Eq.~\eqref{eq_charge_int_along_path}, the charges are evaluated as
\begin{align}
    H[U_{1}] &= \frac{\kappa}{8\pi G}\int \dd y\dd z \sqrt{L(y,z)}W(y,z),\\
    H[U_{2}] &= -\frac{1}{8\pi G \kappa}\int \dd y \dd z \sqrt{L(y,z)}R^{A}(y,z)L_{1A}(y,z),
\end{align}
where the references of the charges are chosen so that they vanish at the background metric.

\bibliography{Charge}

\end{document}